\documentclass[twocolumn]{IEEEtran}
\usepackage[dvips]{graphicx}
\usepackage{amsmath,amssymb,color,verbatim}

\def\norm#1{\left\| #1 \right\|}
\newtheorem{definition}{Definition}[section]
\newtheorem{thm}{Theorem}[section]
\newtheorem{proposition}[thm]{Proposition}
\newtheorem{lemma}[thm]{Lemma}
\newtheorem{corollary}[thm]{Corollary}
\newtheorem{exam}{Example}[section]

\def\bea{\begin{IEEEeqnarray}{rCl}} 
\def\eea{\end{IEEEeqnarray}}
\def\beq{\begin{equation}}
\def\eeq{\end{equation}}
\def\bean{\begin{IEEEeqnarray*}{rCl}} 
\def\eean{\end{IEEEeqnarray*}}

\newtheorem{remark}{Remark}[section]

\DeclareMathOperator*{\tr}{Tr}

\providecommand{\abs}[1]{\ensuremath{\left\lvert #1 \right\rvert}}
\providecommand{\norm}[1]{\ensuremath{\left\Vert #1 \right\Vert}}

\newcommand{\Z}{\mathbb{Z}}
\newcommand{\C}{\mathbb{C}}
\newcommand{\R}{\mathbb{R}}

\providecommand{\abs}[1]{\ensuremath{\left\lvert #1 \right\rvert}}
\providecommand{\norm}[1]{\ensuremath{\left\Vert #1 \right\Vert}}

\title{ Remarks on criteria for achieving  the  optimal diversity-multiplexing gain trade-off } 

\author{Roope Vehkalahti, \emph{Member, IEEE}
\thanks{The research of R. Vehkalahti is funded by   Academy of Finland  grant  \#252457.}
\thanks{R. Vehkalahti is with the Department of Mathematics, FI-20014, University of Turku, Finland  (e-mail: roiive@utu.fi)}}
\begin{document}

\maketitle
\begin{abstract}
In this short note we will prove that non-vanishing determinant (NVD) criterion is not enough for an asymmetric space-time block  code (STBC) to achieve the optimal diversity-multiplexing gain trade-off (DMT). This result is in contrast to the recent  result made by Srinath and Rajan.  In order to clarify the  issue  further the approximately  universality criterion by Tavildar and Viswanath is  translated into language of lattice theory and some conjectures are given.
\end{abstract}

\section{Introduction}
When the diversity-multiplexing gain was introduced in 2003 in \cite{ZT} by Zheng and Tse, the only  code they could prove to achieve the optimal DMT was the Alamouti code. Even the Alamouti code achieved the optimal curve only when it was received with a single antenna. In \cite{EKPKL} Elia et al. proved, when translated to lattice terms, that for a  $2n^2$-dimensional lattice code in $M_n(\C)$ to achieve the optimal diversity-multiplexing gain trade-off it is enough that the code has the non-vanishing determinant property. They also pointed out that the division algebra based codes such as the perfect codes \cite{BORV} are DMT optimal and gave a general construction for DMT achieving  $2n^2$-dimensional lattice codes in $M_n(\C)$.

At the same time Tavildar and Viswanath \cite{TV} did come up with a more general version of this criterion. On lattice theoretic language their result  states that if the products of  the smallest  $n_r$ singular values of any non-zero matrix in  $2n_rn_t$-dimensional lattice $L\subset M_n(\C)$ stays above some fixed constant, then  $L$ achieves the optimal DMT curve in the $n_t\times n_r$ MIMO channel. In the case where $n_r=n_t$, this criterion coincides with the NVD condition.

The NVD condition considers only the case where the lattice $L$ is $2n^2$-dimensional in $M_n(\C)$. However, in  the scenario where $n_t>n_r$, from the decoding complexity point of view, it desirable to use  lattice space-time codes that are at maximum  $2n_tn_r$-dimensional.  Less than $2n_tn_r$-dimensional lattice, on the other hand would, be waste of receiving signal space and energy. Therefore  $2n_tn_r$-dimensional lattice for the $n_t\times n_r$ MIMO channel would be desirable. We refer to such code as an asymmetric space-time code. In such case the criterion \cite{TV} is currently the only general criterion for achieving the optimal DMT.

However, it seem to be that when $n_t>n_r$ asymmetric codes satisfying this conditions are very rare. It is also known that there are space-time codes that are DMT optimal despite satisfying the criterion in \cite{TV} \cite{VeHoLuLa}. It would therefore be nice   to have an easily applicable DMT criterion, that would not be as strict as the criterion in \cite{TV}.

Very recently  in \cite{SriRa} the authors claimed, when translated into lattice theoretic language,  that any $2n_rn_t$-dimensional lattice code $L$ in $M_{n_t\times n_t}(\C)$ does achieve the optimal DMT, if the code has the NVD property. This result would have proved large classes of space-time codes to be DMT-optimal.

In this paper we will give a simple example of a code that fulfill their criterion, but is not DMT optimal. This result suggests that, unfortunately, the Theorem 2 in  \cite{SriRa} is incorrect. Based on our examples  and the recent analysis of several NVD codes in \cite{VeLuLu}, we instead conjecture that actually almost all asymmetric codes with minimum delay are suboptimal from the DMT point of view.

Before the counter-example \ref{example} we translate the  criterion of \cite{TV} to lattice theoretic language. 

\subsection{Diversity-multiplexing gain trade-off}

Consider a Rayleigh block fading MIMO channel with $n_t$ transmit and $n_r$ receive antennas. The channel is assumed to be fixed for a block of $T$ channel uses, but vary in an independent and identically distributed (i.i.d.) fashion to vary from one block to another. Thus, the channel input-output relation can be written as 
\begin{equation}\label{channel}
Y=HX +N, 
\end{equation}
where $H \in M_{n_r \times n_t}(\C)$ is the channel matrix and $N\in M_{n_r \times T}(\C) $ is the noise matrix. The entries of $H$ and $N$ are assumed to be  i.i.d. zero-mean complex circular symmetric Gaussian random variables with variance 1. $X \in M_{n_t\times T}(\C)$ is the transmitted codeword, and $\rho$ denotes the signal-to-noise ratio (SNR).

\begin{definition}\label{def:stbc}
A {\em space-time block code} (STBC) $C$  for some designated  signal to noise ratio (SNR) level $\rho$ is a set of $n_t \times T$ complex matrices satisfying the following average power constraint
\begin{equation}\label{powernorm}
\frac{1}{\abs{C}}\sum_{X \in C} \norm{X}_F^2 \leq \rho T . 
\end{equation}
 A coding scheme $\{ C(\rho)\}$ of STBC is a family of STBCs, one at each SNR level. The rate for the code $C(\rho)$ is thus $R(\rho)=\frac{1}{T} \log \abs{C(\rho)}$. 
\end{definition}

 We say the coding scheme $\{C(\rho)\}$ achieves the DMT of {spatial multiplexing gain} $r$ and \emph{diversity gain} $d(r)$ if the rate satisfies
\begin{equation}\label{ratedemand}
\lim_{\rho \to \infty} \frac{R(\rho)}{\log(\rho)} = r,
\end{equation}
and the average error probability is such that
\[
P_e(\rho) \ \doteq \ \rho^{-d(r)},
\]
where by the dotted equality we mean $f(M) \doteq g(M)$ if 
\beq
\lim_{M\to \infty}\frac{\log(f(M))}{\log(M)} = \lim_{M\to \infty}\frac{\log(g(M))}{\log(M)}. \label{eq:dotdefn}
\eeq
Notations such as $\dot\geq$ and $\dot\leq$ are defined in a similar way.

With the above, the main result in \cite{ZT} is the following.
\begin{thm}[DMT \cite{ZT}] \label{thm:DMT}
Let $n_t$, $n_r$, $T$, $\{C(\rho)\}$, and $d(r)$ be defined as before. Then any STBC coding scheme $\{C(\rho)\}$ has error probability lower bounded by
\beq
P_e(r) \ \dot\geq\ \rho^{-d^*(r)} \label{eq:DMT1}
\eeq
or equivalently, the diversity gain
\beq
d(r) \leq d^*(r), \label{eq:DMT2}
\eeq
when the coding is limited within a block of $T$ channel uses. The function of the optimal diversity gain $d^*(r)$, also termed the optimal DMT, is a piece-wise linear function connecting the points $(r,(n_t-r)(n_r-r))$ for $r=0,1,\ldots,\min\{n_t,n_r\}$. 
\end{thm}

\subsection{Matrix Lattices and their coding schemes}\label{latticesection}
  
  In this section we describe how one can turn a matrix lattice $L \subseteq M_n(\C)$ into a coding scheme  that satisfies the rate \eqref{ratedemand} and average energy \ref{powernorm} demands.
 Throughout the paper will consider STBCs with the minimum delay $n_t = T = n$, and therefore these codes live in the space $M_n(\C)$.

\begin{definition}
A {\em matrix lattice} $L \subseteq M_n(\C)$ has the form
$$
L=\Z B_1\oplus \Z B_2\oplus \cdots \oplus \Z B_k,
$$
where the matrices $B_1,\dots, B_k$ are linearly independent over $\R$, i.e., form a lattice basis, and $k$ is
called the \emph{rank}  or the \emph{dimension} of the lattice.
\end{definition}

\begin{definition}\label{def:NVD}
If the minimum determinant of the lattice $L \subseteq M_n(\C)$ is non-zero, i.e. it satisfies
\[
\inf_{{\bf 0} \neq X \in L} \abs{\det (X)} > 0, 
\]
we say that the lattice satisfies the \emph{non-vanishing determinant} (NVD) property.
\end{definition}

Let $\norm{X}_F = \sqrt{\tr(X^\dagger X)}$ denote the Frobenius norm of $X$.

let us now introduce two coding schemes based on a $k$-dimensional lattice $L$ inside $M_{n}(\C)$. 
\begin{definition}[Spherical shaping]

Given a positive real number $M$ we define
$$
L(M)=\{a \in L \;:\; \norm{a}_F \leq M, a\neq {\bf 0} \}.
$$
We will also use the notation
$$
B(M)=\{a  \in M_n(\C) \;:\; \norm{a}_F \leq M \}
$$
for the sphere with radius $M$.
\end{definition}

\begin{definition}[]
The finite single user codes used in the actual transmission are  then of form
$$
L (M) = \left\{ \sum_{i=1}^{k} b_{i} B_{i} | b_{i} \in \Z, -M \leq b_{i} \leq M \right\},
$$
where $M$ is a given positive number.
\end{definition}

The following two results are well known.
\begin{lemma}\label{spherical}
Let $L$ be a  $k$-dimensional lattice in  $M_{n}(\C)$ and
$L(M)$ be defined as above; 
then
\[
|L(M)|= cM^{k}+ O(M^{k-1}),
\]
where $c$ is some  positive  constant, independent of $M$. 
\end{lemma}

In particular, it follows that we can choose real constants $K_1$ and $K_2$ such that
\begin{equation}\label{latticepointbound}
K_1M^k\geq |L(M)|\geq K_2 M^k.
\end{equation}

\begin{lemma}\label{Mclaurin}
Let $L$ be a $k$-dimensional lattice in $M_n(\C)$. Then
\begin{eqnarray*}
 s_2 M^{k+2}\leq \sum_{X\in L(M)} \norm{X}_F^2\leq s_1M^{k+2},
\end{eqnarray*}
where $s_1$ and $s_2$ are constants independent of $M$.
\end{lemma}

With the above, we are now prepared to give a formal definition of a family of space-time lattice codes of finite size.
\begin{definition}
Given the lattice $L \subset M_n(\C)$, a space-time lattice coding scheme associated with $L$ is a collection of STBCs where each member is given by
\begin{equation}\label{codingscheme}
C_L(\rho)=\rho^{(1/2 -\frac{rn}{k})}L\left(\rho^{\frac{rn}{k}}\right)
\end{equation}
for the desired multiplexing gain $r$ and for each $\rho$ level.
\end{definition}

The normalization factor $\rho^{1/2-\frac{rn}{k}}$ in \eqref{codingscheme} is quite clearly enough for meeting the average power constraint \eqref{powernorm}, but one might wonder whether the STBC $C_L( \rho)$ has average power which is considerably lower than the power constraint in \eqref{powernorm}. From  Proposition \ref{Mclaurin} we have
$$
 \sum_{X \in L\left(\rho^{\frac{rn}{k}}\right)} \rho^{1-\frac{2rn}{k}} \norm{X}_F^2 \doteq \rho^{1-\frac{2rn}{k}}(\rho^{rn/k})^{k+2}=\rho\cdot\rho^{rn}.
$$
On the other hand we also have that $|L(\rho^{\frac{rn}{k}})| \doteq\rho^{rn}$ from Proposition \ref{spherical}. Combining the above shows that the code  $C_L( \rho)$ has the correct average power from the DMT perspective, i.e., in terms of the dotted equality.

\section{ Approximately universality}
In this section we shortly review the approximately universality (AU) criterion  \cite{TV} given in 2006 by Tavildar and Viswanath  and translate their results to consider the lattice based coding schemes introduced in the previous section.  In the introduction we referred to results in \cite{TV}, as a criterion for DMT, but AU is  a considerably stronger condition that only implies DMT. In particular a space-time code can be DMT optimal despite not being approximately universal.   We will only concentrate on the criterion given in \cite{TV} as a method to achieve the optimal DMT and will not describe AU more in this paper.

\begin{thm}\label{Visva}
A sequence of codes  of rate $R(C(\rho))$  is approximately universal over the $n\times n_r$ MIMO-channel if and only if, for every pair of code words
\begin{equation}
\lambda_{1}^2\cdots\lambda_m^2\geq\frac{1}{2^{R(\rho) + o(\log{\rho})}}
\end{equation}
where $\lambda_i$ is the smallest singular value of the normalized (by $\sqrt{\rho}$) codeword difference matrix for a pair codewords of $C(\rho)$ and 
$m=min(n_r, n)$.
\end{thm}

Here the notation $o(\log{\rho})$ refers to a function that is  dominated by $\epsilon log(\rho)$ for any $0<\epsilon$.

\begin{definition}
We will refer to the ith smallest singular value of the matrix $X$ with $\lambda_i(X)$
and set
$$
 \Delta_k(X)=\prod_{i=1}^{k}\lambda_i^2(X).
$$
\end{definition}
We can now extend this definition for lattices.
\begin{definition}
Let us suppose that $L$ is a lattice in $M_n(\C)$. We then note
$$
\Delta_r(L)=\mathrm{infimum}\{\Delta_r(X)| X\in L  \,\,X\neq 0 \}.
$$
\end{definition}

\begin{remark}
The  sentence normalized by $\sqrt{\rho}$ refers that each codeword in $C(\rho)$ is divided with $\sqrt{\rho}$.  
 \end{remark}

The result by Tavildar and Viswanath now transforms into the following.
\begin{corollary}\label{NVD}
Let us suppose that $L$ is a $2n_r n$-dimensional lattice code in $M_n(\C)$ and that 
$$
\Delta_{n_r}(L)\neq 0,
$$
then $L$ is approximately universal (and therefore DMT optimal), when received with $n_r$ antennas.
\end{corollary}
\begin{proof}
Let us now assume that we have scaled our lattice so that $\Delta_{n_r}(L)=1$.
The finite codes we now consider are of type $C_L(\rho)=\rho^{1/2-r/2n_r}L(r/2n_r)$.
For  given elements $\rho^{-r/2n_r+1/2}X$ and $\rho^{-r/2n_r+1/2}Y$ in $\rho^{-r/2n_r+1/2}L(r/2n_r)$ and adding the normalization $\sqrt{\rho}$ we  have 
$$
\Delta_{n_r}\left(\rho^{-r/2n_r+1/2}\left(\frac{X-Y}{\sqrt{\rho}} \right)\right)=\rho^{-r}\Delta_{n_r}(X-Y)\geq \rho^{-r}.
$$
The last inequality here follows as $X-Y\in L$ and we assumed that $\Delta_{n_r}(L)=1$.
On the other hand according to equation \eqref{latticepointbound}  we have $\frac{1}{2^{log(|(C_L(\rho))|)}} \leq \frac{1}{2^{log(A\rho^r)}}=\frac{1}{A\rho^r} $ for some constant
$A$ independent of $\rho$. We obviously  have that $A\in 2^{o(log(\rho)}$. 
\end{proof}

In other words we have the following.
\begin{corollary}
If  a coding scheme $C_L(\rho)$  based on $2n_rn_t$-dimensional lattice code $L$ fulfills
$$
\Delta_{n_r}(X)\geq c\rho^{n_t(1- \frac{r}{n_r})},
$$
for any non-zero codeword $X\in C_L(2\rho)$  any $\rho$ and some fixed constant $c$, then the coding scheme $C_L(\rho)$ is approximately universal.
\end{corollary}

\begin{remark}
The reader should note that approximately universality does allow vanishing product of singular values for a lattice code.
However, this vanishing must be in class $2^{o(log(\rho))}$. In particular vanishing with speed $\rho^{-\epsilon}$ is not allowed, with any fixed
$\epsilon$. 
\end{remark}

\begin{exam}
The Alamouti code together with QAM-symbols can be considered as a 4-dimensional lattice  code $L_{alam}\subset M_2(\C)$. For this code 
$\delta_1(L)>0$. Therefore the coding scheme $C_{L_{alam}}(\rho)$ is approximately universal when received with a single antenna.
\end{exam}

\begin{exam}
The division algebra based codes such as the Perfect codes \cite{BORV} and \emph{maximal order codes} \cite{MeKaikki} are $2n^2$-dimensional lattices in 
$M_n(\C)$ and have the NVD property and are therefore DMT optimal.
\end{exam}

However, in general it seem to be that it is extremely difficult to fulfill these conditions.
As a matter of fact we conjecture.
{\bfseries
\begin{itemize}
\item The conditions  of Corollary \ref{NVD}  can be satisfied only in the cases where either $n_r=n_t$ or when $n_t=2$ and  $n_r=1$.
\end{itemize}
}

\section{Failing of the NVD criteria in lower dimensions}
There are several codes that are DMT optimal despite not fulfilling the approximately universality criterion of the previous section.
 For example the diagonal number field codes \cite{DamenBelfi} and many of the fully diverse quasi-orthogonal codes such as those in \cite{SuXia}  are DMT optimal in the $n_t\times 1$ MIMO channel \cite{VeHoLuLa}. Seen as lattice codes, these are $2n$-dimensional lattices  in $M_n(\C)$ and have the NVD property.  However, they are not approximately universal \cite{EJ}.

It is a tempting idea that simply the NVD condition   and $2nn_r$-dimensional lattice $L \subseteq M_n(\C)$ would be enough for the coding scheme $C_L(\rho)$ to be DMT optimal. 

Using the previous notation we can state this in the form.
\begin{equation}\label{weakNVD}
\Delta_n(X)\geq c\rho^{n_t(1- \frac{r}{n_r})},
\end{equation}
for any $X$ in $C_L(\rho)$ and fixed  positive constant $c$.

However, this is not the case. Let us now build such a code for the  $4\times 1$ MISO channel  that  it fulfills the criteria
\eqref{weakNVD}, but is not DMT optimal in this channel.

Let us consider the Golden code $L_{gold}$. One can see it as an $8$-dimensional lattice in $M_2(\C)$. 
 
 As an NVD code the coding scheme where $\rho^{(1/2-2r/8)}L_{gold}(\rho^{2r/8}$.
Let us now consider a coding scheme, where we take Golden code $C_{gold}$ and then transform in into an $8$-dimensional NVD code $diag(L_{gold})$ in $M_4(\C)$ by setting
$$
diag(X,X)=
\begin{pmatrix}
X& \bf{0}\\
\bf{0} &X
\end{pmatrix},
$$
where, $X$ is a codeword of the Golden code and $\bf{0}$ is a $2\times2$ zero matrix. The set $diag(C_{gold})=\{diag(X)\,|\, X\in C_{gold}\}$ is then an 8-dimensional  NVD lattice code in $M_4(\C)$. However, in order to satisfy the energy normalization demands we have to consider
scheme  $\rho^{1/2-4r/8}diag(L_{gold})(\rho^{4r/8})=C_{L_{gold}}(\rho)$.
\begin{proposition}\label{example}
Let us suppose that $C_{gold}$ is the Golden code. Then the scheme
$\rho^{1/2-4r/8} diag(L_{gold})(\rho^{4r/8})$ is not a DMT optimal code in $4\times 1$ MISO channel.
\end{proposition}
\begin{proof}
Let us now suppose that we transmit a codeword $diag(X)$, where 
$$
X=
\begin{pmatrix}
x_1& x_2\\
x_3& x_4
\end{pmatrix}.
$$

Let us  suppose that the channel vector is $h=[h_1, h_2, h_3, h_4]$ and the noise is $n=[n_1,n_2,n_3,n_4]$. We then have that
$$
h\cdot diag(X) +n=
$$
$$
[h_1x_1+h_2x_3, h_1x_2+h_2x_4, h_3x_1+ h_4x_3, h_3x_2+ h_4x_4] +n.
$$
But this is exactly
$$
\begin{pmatrix}
h_1& h_2\\
h_3& h_4
\end{pmatrix}
\begin{pmatrix}
x_1& x_2\\
x_3& x_4
\end{pmatrix}
+
\begin{pmatrix}
n_1& n_2\\
n_3& n_4
\end{pmatrix},
$$
just written differently. We can see that the error performance of  $ diag(L_{Gold})$, when received with a single antenna is exactly that of $L_{gold}$ when received with two antennas.  The DMT for  the coding scheme $\rho^{1/2-2r/8} L_{Gold}(\rho^{2r/8})$ is  the usual one consisting of  lines connecting points $[r,(2-r)(2-r)^+]$ for integer values. However, this is not directly the DMT for $\rho^{1/2-4r/8}diag(L_{gold})(\rho^{2r/8})$. This is due to the fact that for $diag(L_{gold})$ we have $T=4$ and therefore the diversity gain achieved with multiplexing gain $r$ in the $4\times 1$ channel corresponds to  diversity gain $d(2r)$ in the $2\times 2$ channel. We then see that the DMT is presented by a line connecting points of  $[r,(2-2r)(2-2r)^+]$, where $r=0,\frac{1}{2}, 1$. On the other hand the DMT of the $4+1$ MISO channel is simply a straight line between $[0,4]$ and $[1,0]$. 
\end{proof}

This result indeed proves that in the case where the dimension of the lattice is not $2n^2$ the NVD condition is not enough for the code to reach the optimal DMT.  

\begin{remark}
We note that the previous proposition is just an example of a general principle. We can do the same trick for example by using $4\times 4$ perfect code $L_{perf}$, which is a $32$-dimensional lattice in $M_4(\C)$, to produce 32-dimensional lattice code $diag(L_{perf})\subseteq M_{8\times 8}(\C)$ with the NVD property. The corresponding coding scheme $C_{diag(L_{perf})}(\rho)$ then has DMT curve presented by a line connecting points of   $[r,(4-2r)(4-2r)^+]$, where $r=0,\frac{1}{2}, 1,\frac32, 2$, in the $8\times 2$ MIMO channel. This curve is strongly suboptimal.
\end{remark}

In \cite{VeLuLu} the authors studied several division algebra based codes. The results gotten suggest that in fact it is actually quite rare that a division algebra based code is DMT optimal. Based on this study and the previous examples, we make the following conjecture.

{\bfseries
\begin{itemize}
\item There exist DMT optimal lattice codes $L \subseteq M_n(\C)$  of dimension $2nn_r$ in the $n\times n_r$ MIMO channel if and only if $n=n_t$ or  $n_r=1$.
\end{itemize}
}
This conjecture obviously implies almost completely the previous conjecture on approximately universal codes. However, this result is likely  lot harder to prove.

\end{document}